\documentclass{article}
\usepackage{graphicx} 
\usepackage{amssymb}
\usepackage{stmaryrd}
\usepackage{mathrsfs}
\usepackage{graphicx}
\usepackage{amsmath}
\usepackage{amsthm}
\usepackage{caption}
\usepackage{subcaption}
\usepackage{listings}
\usepackage{verbatim}
\usepackage{braket}
\usepackage{subfiles}
\usepackage{bm}
\usepackage{tikz}
\usepackage{verbatim}
\usepackage[colorlinks]{hyperref}
\usepackage{color}
\usepackage{graphicx}

	\newtheorem{thr}{Theorem} 
	\newtheorem{lmm}{Lemma}
	\newtheorem{cor}{Corollary}

	\newtheorem{prp}{Proposition}

  \newcommand{\lem}{\stackrel{\ast}{<}}
	\newcommand{\gem}{\stackrel{\ast}{>}}
	\newcommand{\eqm}{\stackrel{\ast}{=}}

	\newtheorem{dff}{Definition}
	\newtheorem{asm}{Assumption}
	\newtheorem{rmk}{Remark}
	\newtheorem{clm}{Claim}

	\newcommand{\lea}{<^+}
	\newcommand{\gea}{>^+}
	\newcommand{\eqa}{=^+}
		
	\newcommand{\lel}{<^{\log}}
	\newcommand{\gel}{>^{\log}}

    \newcommand\M{\mathbf{M}}
    \newcommand{\Km}{\mathbf{Km}}
    \newcommand{\bb}{\mathbf{bb}}
    \newcommand{\p}{\mathbf{p}}
    \newcommand{\q}{\mathbf{q}}
    \renewcommand{\r}{\mathbf{r}}
     \newcommand\Q{\mathbb{Q}}
  \newcommand\N{\mathbb{N}}
  \newcommand\R{\mathbb{R}}
  \renewcommand\d{\mathbf{d}}
 \newcommand{\K}{\mathbf{K}}
 \newcommand{\m}{\mathbf{m}}
 \newcommand{\Ks}{\mathbf{Ks}}
 \newcommand{\I}{\mathbf{I}}
 \newcommand\BT{\{0,1\}}
 	\newcommand\FS{\BT^*}\newcommand\IS{\BT^\infty}
	
 \newcommand\ceil[1]{{\lceil#1\rceil}}
 \newcommand\floor[1]{{\lfloor#1\rfloor}}
\title{Game Derandomization}
\author{Sam Epstein\footnote{samepst@jptheorygroup.org}}
\date{\today}
\begin{document}
\maketitle
\begin{abstract}
Using Kolmogorov Game Derandomization, upper bounds of the Kolmogorov complexity of deterministic winning players against deterministic environments can be proved. This paper gives improved upper bounds of the Kolmogorov complexity of such players. This paper also generalizes this result to probabilistic games. This applies to computable, lower computable, and uncomputable environments. We characterize the classic even-odds game and then generalize these results to time bounded players and also to all zero-sum repeated games.  We characterize partial game derandomization. But first, we start with an illustrative example of game derandomization, taking place on the island of Crete.
\end{abstract}
\section{The Minatour and the Labyrinth}
 A hero is trapped in a labyrinth, which consists of long corridors connecting to small rooms. The intent of the the hero is to reach the goal room, which has a ladder in its center reaching the outside. The downside is the hero is blindfolded. The upside is there is a minotaur present to guide the hero.

At every room, the minotaur tells the hero the number of corridors $n$ leading out (including the one which the hero just came from). The hero states a number between 1 and $n$ and the minotaur takes the hero to corresponding door. However the hero faces another obstacle, in that the minotaur is trying to trick him. This means the mapping the minotaur uses is a function of all the hero's past actions. Thus if a hero returns to the same room, he may be facing a different mapping than before. This process continues for a very large number of turns. The question is how much information is needed by the hero to find the exit? Using \textbf{Kolmogorov Game Derandomization}, we get the following surprising good news for the hero. Let $c$ be the number of corridors and $d$ be the number of doors in the goal room.

\begin{center}
\textit{The hero can find the exit using $\log (c/d)+\epsilon$ bits.}
\end{center}

The error term $\epsilon$ is logarithmic and also is dependent on the information the halting sequence has about the entire construct, which is negligible except in exotic cases. Assuming the \textbf{Independence Postulate} \cite{Levin84,Levin13}, one cannot find such exotic constructs in the physical world.

The reasoning for this is as follows. Take a random hero who chooses a corridor with uniform probability. Then the hero is performing a random walk on the graph (of the labyrinth). Assuming the number of turns is greater than the graph's mixing time, the probability the hero is at exit at the end is not less than $d/bc$, for some fixed constant $b$. Then the following theorem can be applied. $\K(x)$ is the prefix Kolmogorov complexity of $x$. $\I(x;\mathcal{H})=\K(x)-\K(x|\mathcal{H})$ is the amount of information  the halting sequence $\mathcal{H}$ has about $x$. Note that in Section \ref{sec:nb}, the bounds of the following theorem are improved.

\begin{thr}[\cite{Epstein23}]
\label{thr:detenv}
If probabilistic agent $\mathbf{p}$ wins against environment $\mathbf{q}$ with at least probability $p$, then there is a deterministic agent of Kolmogorov complexity $\lel\K(\mathbf{p}) -\log p + \I(\langle \mathbf{p},\mathbf{q}\rangle;\mathcal{H})$ that wins against $\mathbf{q}$.
\end{thr}
\section{Setup}
We introduce some key tools necessary for this paper. $O_{a_1,\dots,a_n}(1)$ is a constant dependent on parameters $a_1,\dots,a_n$. For $n\in\N$, $\langle n\rangle = 01^n0$. We use $x\lea y$, $x\gea y$ and $x\eqa y$ to denote $x< y+O(1)$, $x+O(1)>y$ and $x=y\pm O(1)$, respectively. Furthermore, ${\lem}f$, ${\gem}f$ denotes $< O(1)f$ and $>f/O(1)$.  In addition, $x\lel y$ and $x\gel y$ denote $x<y+O(\log y)$ and $x+O(\log x) > y$, respectively. $(x0)^-=(x1)^-=x$. We say $x\sqsubseteq y$ if $xz=y$ for some $z\in\{0,1\}^*$. We say $[A]=1$ if mathematical statement $A$ is true, and $[A]=0$, otherwise. 

The function $\m(x)$ is a universal lower-computable semi-measure. For $D\subseteq\FS$, $\m(D)=\sum_{x\in D}\m(x)$. A continuous semi-measure is a function $\omega:\{0,1\}^*\rightarrow [0,1]$ such that $\omega(\emptyset)=1$ and $\omega(x)\geq \omega(x0)+\omega(x1)$.  For prefix free set  $G$, $\omega(G)=\sum_{x\in G}\omega(x)$. Note that it could be that $\omega(x)>\omega(\{x0,x1\})$. $\M$ is a universal lower-computable continuous semi-measure. Thus for prefix free set $G\subset\FS$, lower computable continuous semi-measure $\omega$, $\M(G)\gem \m(\omega)\omega(G)$. Monotone Kolmogorov complexity is $\Km(x) = \min \{\|p\| : U(p)\sqsupseteq x\}$. This is a slightly different convention than the literature, in that the universal Turing machine $U$ must halt. Mutual information between strings is $\I(x:y)=\K(x)+\K(y)-\K(x,y)$. For strings, $x$, $y$, $\I(x;y)=\K(x)-\K(x|y)$.

A probability $P$ over $\N$ is elementary if it has finite support and its range is a subset of $\Q$. Elementary probabilities can be encoded into finite strings or natural numbers. The randomness deficiency of $x\in \N$ with respect to elementary probability $P$ and $y\in\N$ is $\d(x|P,y)=\ceil{-\log P(x)}-\K(x|P,y)$.
\begin{dff}[Stochasticity]
The stochasticity of $x\in\N$ with respect to $y\in\N$ is $\Ks(x|y)=\min\{\K(P|y)+3\log\max\{\d(x|P,y),1\}:P\textrm{ is elementary}\}$.
\end{dff}

\begin{lmm}[\cite{EpsteinOutliers23,Levin16}]
\label{lmm:ksh}
$\Ks(x|y)\lel \I(x;\mathcal{H}|y).$
\end{lmm}

\begin{lmm}[\cite{Epstein22Exm22}]
\label{lmm:consh}
For partial computable function $f$, \\$\I(f(x);\mathcal{H})\lea \I(x;\mathcal{H})+\K(f)$.   
\end{lmm}

A Win/No-Halt game is a series of interactions between an agent $\p$ and an environment $\q$. Each round starts with $\p$ initiating a move, which is chosen out of $\N$ and then $\q$ responds with a number or $\q$ can halt. Agent $\p$ wins if $\q$ halts the game, otherwise the game can continue potentially forever. Thus $\p$ is a function $(\N\times\N)^*\mapsto \N$ and $\q$ is a function $\N\times (\N\times\N)^*\mapsto \N\cup \{\emptyset\}$. Both $\p$ and $\q$ are are assumed to be computable, however, lower computable and uncomputable environments are studied in Section \ref{sec:ue}. Both the agent and environment can be probabilistic in their choice actions. Thus the probabilities of each action are uniformly computable to any degree of accuracy.
\section{Alternative Proof}
In this section we provide an alternative proof to Theorem \ref{thr:detenv} than the one presented in \cite{Epstein23}. This proof relies on the Monotone EL Theorem, whereas the one in \cite{Epstein23} relies on stochastic processes and is much more extensive. The original EL Theorem can be found at \cite{Levin16,EpsteinEL23}.

\begin{thr}[Monotone EL Theorem, \cite{Epstein24}]
\label{thr:monel}
For prefix free set $G$,\\ $\min_{x\in G}\Km(x)\lel -\log \M(G)+\I(G:\mathcal{H})$.
\end{thr}
\begin{thr}
\label{thr:newproof}
    If probabilistic agent $\mathbf{p}$ wins against environment $\mathbf{q}$ with at least probability $p$, then there is a deterministic agent of Kolmogorov complexity $\lel\K(\mathbf{p}) -\log p + \I(\langle \mathbf{p},\mathbf{q}\rangle;\mathcal{H})$ that wins against $\mathbf{q}$.
\end{thr}
\begin{proof}
    On can associate $(\N\times \N)^*$ to $\N$ in the natural way. Thus $\p$ can be seen as a stochastic process over a set of random variables $\{X_i\}_{i=1}^\infty$. So for each $X_i$, the index $i$ encodes a history of interactions with $\q$. So an interaction with $\q$ is a set $\{(a_i,b_i)\}_{i=1}^n$, where each $a_i\in \N$ encodes the history of the interaction and $b_i$ is the agent's action. By encoding each number $n\in\N$ into a string $\langle n\rangle$, $\p$ can be seen as a computable measure over $\{0,1\}^\infty$. With this encoding, an interaction is an open set in the Cantor space. Let $s = \ceil{-\log p}+1$. Let $G\subset\{0,1\}^\infty$ be a clopen set, computable from $\q$, $\p$ and $s$ where each of its elements encode a winning interaction and $\p(G)>2^{-s}$. Thus $\K(G|\p,\q,s)=O(1)$. Using the Monotone EL Theorem \ref{thr:monel},
    There exists $y$ such that $y\sqsupseteq x\in G$ and 
    \begin{align}
    \nonumber
    \K(y) &\lel -\log \M(G) + \I(G;\mathcal{H})\\
    \label{eq:ap1}
    &\lel \K(\p)+s + \I(\langle \p, s,\q\rangle ;\mathcal{H})\\
    \nonumber
    &\lel \K(\p)-\log p + \I(\langle \p,\q\rangle ;\mathcal{H})
    \end{align} 
    Equation \ref{eq:ap1} is due to Lemma \ref{lmm:consh}. This $y\in\FS$ encodes a finite series of numbers, which in turn encodes a series of actions $a\in\N$ at a number of histories $h\in (\N\times\N)^*$. Thus $y$ encodes a deterministic player $\p'$. When paired with $\q$, the deterministic agent $\p'$ will result in a winning interaction.
\end{proof}
\section{Probabilistic Games}
In this section we prove Kolmogorov Game Derandomization over probabilistic environments. This is an extension to Theorem \ref{thr:detenv}, enabling the characterization of all computable and probabilistic environments.

The main proof uses the notion of a \textit{game fragment}. A game fragment $\mathcal{F}$ is a finite tree, where each edge has a number $n\in\N$ representing an action. On the even levels, the edges are coupled with rational weights in $[0,1]$ and the summation of weights on edges with the same parent node is not more than 1. Such fragments $\mathcal{F}$ can be coupled with an probabilistic agent $\p$, who fills in the weights of each odd level edges with its probabilistic action. In such a coupling, the weight of each path is the product of the probabilities along each edge of the path. $\mathrm{Weight}(\p,\mathcal{F})$ is the sum of the weights of each path. 
\begin{clm}
\label{clm}
If there is a fragment $\mathcal{F}$ where each path from the root to a leaf represents a winning interaction with an environment $\q$ and the weights of $\mathcal{F}$ are not more than $\q$'s probabilities of those particular actions, then $\mathrm{Weight}(\p,\mathcal{F})$ is not more than the probability that $\p$ wins against $\q$.
\end{clm}
\begin{thr}[Probabilistic Environments]
\label{thr:main}
    Let $\p$ be a probabilistic agent and $\q$ be a probabilistic environment. If $\p$ wins in the Win/No-Halt game against $\q$ with probability $> 2^{-s}$, $s\in\N$, then there is a deterministic agent of complexity $\lel \K(\p)+2s +\I(\langle \p,\q\rangle;\mathcal{H})$ that wins with probability $> 2^{-s-1}$.
\end{thr}
\begin{proof}
    We relativize the universal Turing machine to $\langle \p,s\rangle$. Thus this information is on an auxiliary tape and implicitly in the conditional of all complexity terms.  Let $\mathcal{F}$ be a game fragment corresponding the environment $\q$ such that each path is a winning interaction and $\mathrm{Weight}(\p,\mathcal{F})>2^{-s}$ and also $\K(\mathcal{F}|\q)=O(1)$.
    Note that the actions of $\mathcal{F}$ are rationals which lower bound $\q$'s computable action probabilities.
    Let $Q$ be an elementary probability measure that realizes $\Ks(\mathcal{F})$ and $d=\max\{\d(\mathcal{F}|Q),1\}$. Without loss of generality, one can limit the support of $Q$ to encodings of game fragments $\mathcal{G}$ such that $\mathrm{Weight}(\mathcal{G},\p)>2^{-s}$. This can be done by defining a new probability $Q'$ that is $Q$ conditioned on the above property, which is straightforward but tedious. Let $m$ be the longest path and $\ell$ be the largest action number of any game fragment in the support of $Q$. We define a probability $P$ over deterministic agents $\mathbf{g}$ defined up to $m$ steps and up to $\ell$ actions. Each action of the deterministic agent is determined by the corresponding probability of actions in that turn by $\p$. Using backwards induction, for each math fragment $\mathcal{G}$ in the support of $Q$, $$\mathbf{E}_{\mathbf{g}\sim P}[\mathrm{Weight}(\mathbf{g},\mathcal{G})]>2^{-s}.$$
Let $N$ be a number to be specified later. Assume we randomly define $N$ determinstic agents $\{\mathbf{g}_i\}_{i=1}^N$, each drawn i.i.d. from $P$. For math fragment $\mathcal{G}$ in the support of $Q$, $X_\mathcal{G}=\frac{1}{N}\sum_{i=1}^N \mathrm{Weight}(\mathbf{g_i},\mathcal{G})$. Each such $X_\mathcal{G}$ is a random variable. By the Hoeffding's inequality,
$$\mathrm{Pr}(X_\mathcal{G}\leq 2^{-s-1}) < 2\mathrm{exp}(-N2^{-2s-2}).$$
Let $N=d2^{2s+3}$. Then it is possible to find a set of $N$ deterministic agents such that
$$Q(\{\mathcal{G}:X_\mathcal{G}\leq 2^{-s-1}\}) < e^{-d}.$$
In the above formula, each $X_\mathcal{G}$ is a fixed value and no longer a random variable. It must be that $X_\mathcal{F}>2^{-s-1}$. Otherwise using $Q$-test $t(\mathcal{G})=[X_\mathcal{G}\leq 2^{-s-1}]e^d$,
$$1.44d<\log t(\mathcal{F}) \lea \d(\mathcal{F})\eqa d.$$
This is a contradiction for large enough $d$ which we can assume without loss of generality. Thus since $X_\mathcal{F}>2^{-s-1}$ there exists deterministic agent $\mathbf{g}_i$ such that $\mathrm{Weight}(\mathbf{g}_i,\mathcal{F})> 2^{-s-1}$. Thus, by Claim \ref{clm}, $\mathbf{g}_i$ wins against $\q$ with probability more than $2^{-s-1}$. So,
\begin{align}
\nonumber
    \K(\mathbf{g}_i|s,\p) &\lea \log N + \K(\{\mathbf{g}_i\}|s,\p)\\
\nonumber
    \K(\mathbf{g}_i) &\lea \K(s,\p)+2s+\log d+\K(d,Q|s,\p)\\
\nonumber
    &\lel \K(\p)+2s+3\log d +\K(Q|s,\p)\\
\label{eq2}
    &\lel \K(\p)+2s+\Ks(\mathcal{F}|s,\p)\\
\label{eq25}
    &\lel \K(\p)+2s+\Ks(\mathcal{F})+O(\log \K(s,\p))\\
\label{eq3}
    &\lel \K(\p)+2s+\I(\mathcal{F};\mathcal{H})\\
\label{eq4}
    &\lel \K(\p)+2s+\I(\langle \p,s,\q\rangle;\mathcal{H})\\
\label{eq5}
    &\lel \K(\p)+2s+\I(\langle \p,\q\rangle;\mathcal{H}).
\end{align}

Equations \ref{eq2} and \ref{eq25} follow from the definition of stochasticity, $\Ks$. Equation \ref{eq3} follows from Lemma \ref{lmm:ksh}. Equation \ref{eq4} follows from Lemma \ref{eq4} and the fact that $\mathcal{F}$ is computable from $\p$, $\q$, and $s$. Equation \ref{eq5} is due to the logarithmic precision of the inequality and $\K(s)=O(\log s)$
\end{proof}
\begin{cor}
\label{cor:main}
Let $\epsilon\in(0,1)$
 be computable. Let $\p$ be a probabilistic agent and $\q$ be a probabilistic environment. If $\p$ wins in the Win/No-Halt game against $\q$ with probability $> 2^{-s}$, $s\in\N$, then there is a deterministic agent of complexity $\lel \K(\p)+2s +\I(\langle \p,\q\rangle;\mathcal{H})+O_\epsilon(1)$ that wins with probability $> \epsilon2^{-s}$.
\end{cor}
\begin{cor}
\label{cor:main2}
Let $\epsilon\in(0,1)$ and $p\in(0,1)$ both be computable. Let $\p$ be a probabilistic agent and $\q$ be a probabilistic environment. If $\p$ wins in the Win/No-Halt game against $\q$ with probability $> p$, then there is a deterministic agent of complexity $\lel \K(\p)+\I(\langle \p,\q\rangle;\mathcal{H})+O_{p,\epsilon}(1)$ that wins with probability $> \epsilon p$.
\end{cor} \newpage

 \section{Computability of Environments}
\label{sec:ue}
\begin{clm}
Theorem \ref{thr:main} and Corollaries \ref{cor:main} and \ref{cor:main2} also apply to probabilistic environments $\q$ with lower computable probabilities since $\K(\mathcal{F}|\p,s,\q)=O(1)$, where $\langle \p,s,\q\rangle$ consists of a program to compute $\p$, the number $s$, and a program to lower compute $\q$. This is because one lower enumerates the probabilities of $\q$ until one can find the corresponding game fragment $\mathcal{F}$ such that $\mathrm{Weight}(\p,\mathcal{F})>2^{-s}$.
\end{clm}
 In this section, we derive the results of Theorems \ref{thr:detenv} and \ref{thr:main} with respect to uncomputable environments. We will use the following mutual information term between infinite sequences.
\begin{dff}[\cite{Levin74}]
\label{dff:infinf}
For $\alpha,\beta\in\{0,1\}^\infty$,\\ $\I(\alpha:\beta)=\log \sum_{x,y\in\{0,1\}^*}\m(x|\alpha)\m(y|\beta)2^{\I(x:y)}$.
\end{dff}
\begin{prp}
\label{prp}
    $\I(x;\mathcal{H})\lea \I(\alpha:\mathcal{H})+\K(x|\alpha)$.
\end{prp}
Now a probabilistic environment $\q$ is of the form $\N\times(\N\times\N)^*\rightarrow [0,1]$. We fix a computable function $\ell$ such that for every environment $\q$ there is an infinite sequence $\alpha$ such that $\ell(\alpha,\cdot)$ computes $\q$. Let $\ell[\q]$ be the set of all such infinite sequences $\alpha$.
\begin{dff}
\label{dff:uncompagent}
    For probabilistic environment $\q$, $\I(\q:\mathcal{H})=\inf_{\alpha\in\ell[\q]}\I(\alpha:\mathcal{H})$.
\end{dff}
\begin{thr}[Uncomputable Environments]
\label{thr:uncompprob}
    Let $\p$ be a probabilistic agent and $\q$ be a (potentially uncomputable) probabilistic environment. If $\p$ Wins in the Win/No-Halt game against $\q$ with probability $> 2^{-s}$, $s\in\N$, then there is a deterministic agent of complexity $\lel \K(\p)+2s +\I(\langle \p,\q\rangle:\mathcal{H})$ that wins with probability $> 2^{-s-1}$.
\end{thr}
\begin{proof}
    Using $\p$, $s$, any encoding $\alpha\in\ell[\q]$, and $\ell$, one can construct the math fragment $\mathcal{F}$ described in the proof of Theorem \ref{thr:main}. Let $\alpha\in\ell[\q]$ and $\I(\alpha:\mathcal{H})<\I(\q:\mathcal{H})+1$. Thus $\K(\mathcal{F}|\p,s,\alpha)=O(1)$. Using Proposition \ref{prp}, the definition of $\I(\q:\mathcal{H})$, and the reasoning of the proof of Theorem \ref{thr:main}, this theorem follows.
\end{proof}

The follow Theorem extends Theorem \ref{thr:detenv} to uncomputable environments. 
\begin{thr}
If probabilistic agent $\mathbf{p}$ wins against deterministic, and potentially uncomputable, environment $\mathbf{q}$ with at least probability $p$, then there is a deterministic agent of complexity $\lel\K(\mathbf{p}) -\log p + \I(\langle \mathbf{p},\mathbf{q}\rangle:\mathcal{H})$ that wins against $\mathbf{q}$.
\end{thr}
\begin{proof}
This follows from using the same reasoning as the proof for Theorem \ref{thr:uncompprob} and the proof of Theorem \ref{thr:detenv} found in \cite{Epstein23}. 
\end{proof}

\section{Revisiting Crete}
\label{sec:crete}
Suppose the minatour has gotten fed up with the hero, who can find the exit using a very small amount of information. The minatour decides to use chance to his advantage. At every room the minatour computes a probability over all possible mappings of numbers to doors and selects a mapping at random. This probability is a functions of all the hero's previous actions. However due to derandomization of probabilistic games, the hero can achieve the following results. Let $c$ be the number of corridors and $d$ be the number of doors in the goal room.
\begin{thr}
     Given any labyrinth and probabilistic minatour $(L,M)$, there is a deterministic hero $\p$ of complexity $\K(\p)\lel 2\log (c/d)+\I(\langle L,M\rangle;\mathcal{H})$ that can find the goal room with probability greater than $d/4.03c$.
\end{thr}
\begin{proof}
We assume the number of turns $N$ is long enough such that given a random walk of $N$ steps, the probability of being a the goal room is $>d/2.01c$. Let $s\in\N$ be the largest number such that $2^{-s}<d/2.01c$. Applying Corollary \ref{cor:main}, where $\epsilon = 4.02/4.03$ and where the probabilistic agent $\p$ chooses each door with uniform probability, on gets a deterministic agent of complexity $\lel \K(\p)+2s+\I(\langle \p,s,(L,M)\rangle;\mathcal{H})+O_\epsilon(1)\lel 2\log(c/d)+\I(\langle L,M\rangle;\mathcal{H})$. This agent wins with probability $>\epsilon 2^{-s}\geq \epsilon d/4.02c\geq d/4.03c$. 
\end{proof}

\section{{\sc Even-Odds}}
\label{sec:evenodds}
We define the following game, entitled {\sc Even-Odds}. There are $N$ rounds. The player starts out with a score of 0. At the start of each round, the environment $\q$ secretly records a bit $e_i\in\{0,1\}$. The player sends $\q$ a bit $b_i$ and the environment responds with $e_i$. The agent gets a point if $e_i\oplus b_i=1$. Otherwise the agent loses a point. The environment $\q$  can be any probabilistic algorithm. 

\begin{thr}
\label{thr:EvenOdds}
For large enough number of rounds, $N$, given any probabilistic environment $\q$ there is a deterministic agent $\p$ of complexity $\K(\p)\lel \I(\q;\mathcal{H})$ that can achieve a score of $\sqrt{N}$ with probability $>1/6.5$.
\end{thr}
\begin{proof}
We describe a probabilistic agent $\p'$. At round $i$, $\p'$ submits 0 with probability 1/2. Otherwise it submits 1. By the central limit theorem, for large enough $N$, the score of the probabilistic agent divided by $\sqrt{N}$ is $S\sim\mathcal{N}(0,1)$. So the probability that a probabilistic agent gets a score of at least $\sqrt{N}$ is greater than $1/6.4$. One can construct a Win/No-Halt game where the player $\p'$ wins if it has a score of at least $\sqrt{N}$. Thus $\p'$ wins with probability greater than $p=1/6.4$. Thus by Corollary \ref{cor:main2}, with $\epsilon = 6.4/6.5$, there exists a deterministic agent $\p$ that can beat $\q$ with complexity
\begin{align*}
\K(\p)&\lel \K(\p')+\I(\langle \p',\q\rangle;\mathcal{H})\lel \I(\q;\mathcal{H}).
\end{align*}
Furthermore $\p$ wins with probability $> 1/6.5$.
\end{proof}

\section{Resource Bounded {\sc Even-Odds}}
So far the results have been on resource-free Kolmogorov complexity; the time to construct the players is not taken into account. However the results of this paper can be reinterpreted in terms of resource-bounded Kolmogorov complexity. Thus an environment playing {\sc Even-Odds} in polynomial time implies the existence of a player that can be constructed in polynomial time that achieves a certain score. I could not find results in the literature with regard to resource bounded Kolmogorv complexity of players of games. 
\begin{dff}[Time Bounded Kolmogorov Complexity]
For $x\in\FS$,
    $\K^m(x)=\{\|y\|:U(y)=x \textrm{ in }m\textrm{ steps}\}$. For deterministic player $\p$, $\K^m(\p) = \min\{\K^m(x):U_x(\cdot) \textrm{ computes }\p\textrm{ in $O(n)$ time}\}$.
\end{dff}
\begin{asm}
\label{asm}
\textbf{Crypto} is the assumption that there exists a language in $\mathbf{DTIME}(2^{O(n)})$ that does not have size $2^{o(n)}$ circuits with $\Sigma_2^p$ gates. 
\end{asm}
\begin{dff}
$\mathbf{FP}'=\{f:f\in\mathbf{FP}\textrm{ and }\|x\|=\|f(x)\|\}$.
\end{dff}
\begin{dff}
For $\zeta\in\mathbf{FP}'$ we say that $\zeta$ samples $D\subset\BT^N$ with probability $\gamma$ if $|\BT^N\cap \zeta^{-1}(D)|/2^N>\gamma$.
\end{dff}
\begin{dff}
    For $L\subseteq\FS$, $L_N=\{ x: x\in L,\|x\|=N\}$.
\end{dff}
\begin{thr}[Resource EL Theorem, \cite{Epstein23}]
\label{thr:resel}
Assume \textbf{Crypto}. Let $L\in\mathbf{P}$, $\zeta\in\mathbf{FP}'$, and assume $\zeta$ samples $L_N$ with probability $\delta_N$. Then for some polynomial $p$,
$$\min_{x\in L_N}\mathbf{K}^{p(N)}(x) < -\log\delta_N+ O(\log N).$$
\end{thr}
\begin{rmk}
    Due to \cite{AntunesFo09,LuOlZi22}, the polynomial function $p$ is a polynomial function of the running times of $\zeta$ and the algorithm that decides $L$. Furthermore this polynomial has a small degree and small coefficients.
\end{rmk}
\begin{thr}
\label{thr:Kpp}
Assume \textbf{Crypto}. Let $\q$ be a deterministic polynomial time computable environment that plays {\sc Even-Odds} continuously. There is a polynomial function $p$ where for large enough $N\in\N$, there exists a deterministic agent $\p$ that can achieve a score of $\sqrt{N}$ after $N$ turns, with $\K^{p(N)}(\p)=O(\log N)+O_\q(1)$.
\end{thr}
\begin{proof}
We associate each $x\in\BT^{2N}$ with an interaction of $N$ turns of {\sc Even-Odds}. The odd bits are the agent's actions and the even bits are $\q$'s deterministic actions. Let $L\subset\FS$ be all the interactions $x\in\BT^{2N}$ with $\q$ that results in a score greater than $\sqrt{N}$. Since $\q$ is polynomial time computable, $L\in P$. For large enough $N$, by the central limit theorem, $|L_N|2^{-N}>1/6.4$. We define a function $\zeta\in\mathbf{FP}'$, such that every odd bit is preserved and every outputted even bit is $\q$'s response to the previous even bits (agent's actions). So for large enough $N$, $\zeta$ samples $L_N$ with probability $>1/6.4$. So by Theorem \ref{thr:resel}, there exists a deterministic player $\p$ that can achieve a score greater than $\sqrt{N}$ and a polynomial function $p$ independent of $N$, for large enough turns $N$, with $$\K^{p(N)}(\p) < -\log 1/6.4 +  O(\log N)+O_\q(1)<O(\log N)+O_\q(1).$$
\end{proof}

\section{Zero-Sum Repeated Games}

In this section we generalize from the {\sc Even-Odds} game of Section \ref{sec:evenodds} to all zero-sum repeated games. A simultaneous game between two players $A$ and $B$ are defined as follows. At each turn, both players simultaneously play an action $a,b\in\N$. Each action is a function of the previous turns. Thus both $A$ and $B$ are of the form $(\N\times\N)^*\rightarrow \N$. This process continues of $N$ turns. The determination of the outcome after $N$ turns is dependent on each such game.

A \textit{Zero-Sum Repeated Game} is when the simultaneous game is a series of identical zero-sum stage games $\mathcal{G}$. The payoffs of the stage game $\mathcal{G}$ are assumed to be rationals. Each player starts with a score of 0. A zero sum stage game is when the actions of $A$ and $B$ are chosen from $\{1,\dots,n\}$. After the actions occur, each player is given a penalty or a prize. The total prizes and penalties for each player sum to 0. 
\begin{quote}
\textit{This section characterizes uncomputable, (lower)computable, and also polynomial-time computable players that are either deterministic or probabilistic.}
\end{quote}

\begin{dff}
    The $\mathcal{K}(\mathcal{G})$ constant of a zero-sum stage game $\mathcal{G}$ is equal to $c\in\R_{\geq 0}$ where the total probabilistic mass of $\mathcal{N}(0,\sigma^2)$ after $c$ is $=1/3$ and where $\sigma^2$ is the variance of the payoffs of player $A$ when playing uniforming randomly.
\end{dff}
Thus, as the prizes and penalties of the stage game $\mathcal{G}$ increases, the corresponding constant $\mathcal{K}(\mathcal{G})$ also increases.
\begin{thr}
    For repetition of zero-sum stage game $\mathcal{G}$ that has $n$ actions, over large enough turns $N$, for all computable deterministic players $B$, there is a computable deterministic player $A$ that can achieve a score greater than $\mathcal{K}(\mathcal{G})\sqrt{N}$ with complexity $\K(A) \lel \K(n)+\I(\langle B,N,\mathcal{G}\rangle;\mathcal{H})$.
\end{thr}
\begin{proof}
Let player $A'$ play each action with uniform probability at every turn. So $\K(A')=n$. At every turn, its payoff will be a random variable $X_i$ with 0 mean and $\sigma^2$ variance. By large enough $N$, by the central limit theorem, the score of $A'$ divided by $\sqrt{N}$ is distributed according to $\mathcal{N}(0,\sigma^2)$. Thus with probability $1/3$, $A'$ will have a score greater than $\mathcal{K}(\mathcal{N})\sqrt{N}$.  One can turn $N$ rounds of play into a Win/No-Halt game where the environment has complexity $\K(B,N,\mathcal{G})$ and $A'$ wins if it has a score greater than $\mathcal{K}(\mathcal{G})\sqrt{N}$. Thus by Theorem \ref{thr:detenv}, there is a deterministic player $\mathcal{A}$ that achieves a score greater than $\mathcal{K}(\mathcal{N})\sqrt{N}$ and $\K(A)\lel \K(n)+\I(\langle B,N,\mathcal{G}\rangle;\mathcal{H}).$
\end{proof}

\begin{thr}
    For repetition of zero-sum stage game $\mathcal{G}$ that has $n$ actions, over large enough turns $N$,  for all computable probabilistic players $B$, there is a computable deterministic player $A$  with complexity $\K(A) \lel \K(n) +\I(\langle B,N,\mathcal{G}\rangle;\mathcal{H})$ that can achieve a score greater than $\mathcal{K}(\mathcal{N})\sqrt{N}$ with probability $>1/4$.
\end{thr}
\begin{proof}
Let player $A'$ play each action with uniform probability at every turn. So $\K(A')=n$. At every turn, its payoff will be a random variable $X_i$ with 0 mean and $\sigma^2$ variance. By large enough $N$, by the central limit theorem, the score of $A'$ divided by $\sqrt{N}$ is distributed according to $\mathcal{N}(0,\sigma^2)$. Thus with probability $>1/3.5$, $A'$ will have a score greater than $\mathcal{K}(\mathcal{N})\sqrt{N}$. One can turn $N$ rounds of play into a Win/No-Halt game where the environment has complexity $\K(B,N,\mathcal{G})$ and $A'$ wins if it has a score greater than $\mathcal{K}(\mathcal{G})\sqrt{N}$. Thus by Corollary \ref{cor:nb3}, with $p=1/3.5$ and $\epsilon=3.5/4$ we get deterministic player $A$ of complexity $\K(A)\lel \K(n)+ \I(\langle B,N,\mathcal{G}\rangle;\mathcal{H})$ that wins probability $>1/4$.
\end{proof}
The following theorem covers uncomputable opponents of zero-sum repeated games. The information term $\I(\cdot:\mathcal{H})$ is from Definition \ref{dff:uncompagent}.
\begin{thr}
    For repetition of zero-sum stage game $\mathcal{G}$ that has $n$ actions, over large enough turns $N$, for all uncomputable deterministic players $B$, there is a computable deterministic player $A$ that can achieve a score greater than $\mathcal{K}(\mathcal{G})\sqrt{N}$ with complexity $\K(A) \lel \K(n)+\I(\langle B,N,\mathcal{G}\rangle:\mathcal{H})$.
\end{thr}
\begin{proof}
    This follows from the same reasoning as the proof of Theorem \ref{thr:uncompprob}.
\end{proof}

\begin{thr}
Assume \textbf{Crypto}. Let $B$ be deterministic polynomial time computable agent that plays zero-sum stage game $\mathcal{G}$ continuously. Assume $\mathcal{G}$ has $2^n$ actions. There is a polynomial function $p$ where for large enough $N\in\N$, there exists a deterministic agent $A$ that can achieve a score of $\mathcal{K}(\mathcal{G})\sqrt{N}$ after $N$ turns, with $\K^{p(N)}(A)=O(\log N)+O_B( 1)$.
\end{thr}
\begin{proof}
We associate each $x\in\BT^{2nN}$ with an interaction of $N$ turns of {\sc Even-Odds}. Each odd segment of $n$ bits are the agent's actions and the even segments of $n$ bits are $\q$'s deterministic actions. Let $L\subset\FS$ be all the interactions $x\in\BT^{2nN}$ with $B$ that results in a score greater than $\mathcal{K}(\mathcal{G})\sqrt{N}$. Since $B$ is polynomial time computable, $L\in P$. For large enough $N$, by the central limit theorem and the definition of $\mathcal{K}(\mathcal{G})$, for large enough $N$, $|L_N|2^{-N}>1/4$. We define a function $\zeta\in\mathbf{FP}'$, such that every odd segment of $n$ is preserved and every outputted even segment of $n$ bits is $B$'s response to the previous even digits (agent's actions). So for large enough $N$, $\zeta$ samples $L_N$ with probability $>1/4$. So by Theorem \ref{thr:resel}, there exists a deterministic player $A$ that can achieve a score greater than $\mathcal{K}(\mathcal{G})\sqrt{N}$ in large enough $N$ turns and a polynomial $p$ independent of $N$ with $$\K^{p(N)}(A) < -\log 1/4 +  O(\log N)+O_B(1)<O(\log N)+O_B(1).$$
\end{proof}

\begin{rmk}
Though this paper provides a resource bounded derandomization of zero-sum repeated games, what's noticeably absent is general resource bounded derandomization of arbitrary Win/No-Halt games. This absence is because the Resource EL Theorem \ref{thr:resel} is a statement about infinite set of strings $\{D_n\}_{n=1}^\infty$, where $D_n\subseteq\BT^n$. This form is compatible with zero-sum repeated games for any of the $n$ rounds but not so much for games like the one in Section \ref{sec:crete}, which would require an infinite set of graphs and door mappings. Furthermore, the number of possible actions is different for each turn, which is in conflict with the requirement of the Resource EL Theorem for sets of strings of the same length.
\end{rmk}

\section{New Bounds}
\label{sec:nb}
In fact, as shown in this section, the bounds of Theorems \ref{thr:detenv} and \ref{thr:main} can be improved. These new bounds use a summation of all probabilistic agent's winning probabilities, weighted by their algorithmic probabilities. A semi-agent is a probabilistic player whose action probabilities don't necessarily sum to 1 at each turn. The leftover probability represents the chance that the semi-agent freezes and does not win. A semi-agent $\p$ is modelled by the function $(\N\times\N)^*\times \N \rightarrow [0,1]$. A semi-agent $\p$ is lower computable if the probability of $\p$ performing a finite set of actions given the environment's action is lower computable. More formally, for every action $a\in\N$ of the player and response $e\in\N$ of the environment and history $h\in(\N\times\N)^*$, $\p(a,h)\in[0,1]$, $p(a,h)$ is lower computable, and $\sum_{b\in \N}\p(b,(h,(a,e))\leq \p(a,h)$. Also note that, in general, a lower computable semi-agent will not have (even lower) computable probabilities of an action given its history. Let $\mathrm{Win}(\p,\q)$ is the winning probability of semi-agent $\p$ against environment $\q$. If $\p$ is lower computable, then $\mathrm{Win}(\p,\q)$ is lower computable. The algorithmic probability of a semi-agent is $\m(\p)=\sum\{2^{-\|\ell\|}:U_\ell(\cdot)\textrm{ lower computes }\p\}$. Thus if $\p$ is not lower computable, then $\m(\p)=0$. Let $\mathcal{A}$ be the set of all lower computable semi-agents. The set $\mathcal{A}$ can be enumerated in the standard way of algorithmic information theory.

\begin{dff}
    For (deterministic or probabilistic) environment $\q$, let $\xi(\q)=\ceil{-\log \sum_{\p\in\mathcal{A}}\m(\p)\mathrm{Win}(\p,\q)}$ is a score of how hard it is for each semi-agent to win against $\q$, weighted by its algorithmic probability.
\end{dff}
\begin{rmk}
Let $\Omega = \sum\{2^{-\|p\|}:U(p)\textrm{ halts}\}$ be Chaitin's Omega and $\Omega^t = \sum\{2^{-\|p\|}:U(p)\textrm{ halts in time $t$}\}$. For a string $x$, let $BB(x)=\min \{t:\Omega^t>0.x+2^{-\|x\|}\}$. Note that $BB(x)$ is undefined if $0.x+2^{-\|x\|}>\Omega$. For $n\in \N$, let $\bb(n) = \max\{BB(x): \|x\|\leq n\}$. $\bb^{-1}(m) = \arg\min_n \{\bb(n-1)<m\leq \bb(n)\}$. Let $bb(n)=\arg\max_x\{BB(x) :\|x\|\leq n\}$.
\end{rmk}
\begin{lmm}
\label{lmm:rec}
For $n=\bb^{-1}(m)$, $\K(bb(n)|m,n)=O(1)$.
\end{lmm}
\begin{proof}
Enumerate strings of length $n$, starting with $0^n$, and return the first string $y$ such that $BB(y)\geq m$. This string $y$ is equal to $bb(n)$, otherwise $BB(y^-)$ is defined and $BB(y^-)\geq BB(y)\geq m$. Thus $\bb(n-1)\geq m$, causing a contradiction.
\end{proof}
\begin{prp}$ $\\
\label{prp:bb}
\vspace*{-0.5cm}
    \begin{enumerate}
    \item $\K(bb(n))\gea n$.
    \item $\K(bb(n)|\mathcal{H})\lea \K(n)$.
    \end{enumerate}
\end{prp}
\begin{thr}[Derandomization, Improved Bounds]
\label{thr:nb1}
Let $\q$ be a computable deterministic environment. There is a deterministic agent of complexity\\ $\lel \xi(\q)+\I(\q;\mathcal{H})$ that wins against $\q$.
\end{thr}
\begin{proof}
    Let $\p$ be the universal lower-computable semi-agent, where $$\p=\sum_{\p'\in\mathcal{A}}\m(\p')\p'.$$ Let $s=\ceil{-\log\mathrm{Win}(\p,\q)}+1=\xi(\q)+1$. Let $\p^c$ be $\p$ after enumerating $c$ steps. Let $m$ be the smallest number such that
    $\mathrm{Win}(\p^m,\q)>2^{-s-1}$. Let $n=\bb^{-1}(m)$, $k=\bb(n)$. Let $\p^k$ be $\p$ after enumerating $k$ steps, and completed to an agent (from a semi-agent). Thus $\mathrm{Win}(\p^k,\q)>2^{-s-1}$. Let $b=bb(n)$. Theorem \ref{thr:detenv}, conditioned on $b$ gives a deterministic agent $\r$ such that 
    \begin{align}
        \nonumber
        \K(\r|b)&\lel \K(\p^k|b)+s+\I(\langle \p^k,\q\rangle;\mathcal{H}|b)\\
        \nonumber
        &\lel s+\I(\q;\mathcal{H}|b)\\
        \nonumber
       \K(\r) &\lel s+\K(b)+\K(\q|b)-\K(\q|b,\mathcal{H})\\
       \label{eq:nb1}
       &\lel s+\K(b,\q)+\K(\K(b))-\K(\q|b,\mathcal{H})\\
       \nonumber
        &\lel s+\K(\q)+\K(b|\q,s)-\K(\q|\mathcal{H})\\
        \nonumber
        &\lel s+\K(\q)+\K(b|\q,s,n)+\K(n)-\K(\q|\mathcal{H})\\
        \label{eq:nb2}
        &\lel s+\K(\q)-\K(\q|\mathcal{H}).
    \end{align}
    The equations follow from Proposition \ref{prp:bb}. Equation \ref{eq:nb1} follows from the chain rule. Equation \ref{eq:nb2} follows from Lemma \ref{lmm:rec}.
\end{proof}
The following corollary is an update to Theorem 4 in \cite{EpsteinBe2011} which is derived from Theorem 6 in \cite{VereshchaginVi05}. Another form of this theorem can be found in Lemma 6 in \cite{VereshchaginVi04}. It also follows from \cite{Levin16,EpsteinEL23}. In fact, theorems of this general form are ubiquitous in algorithmic information theory. It says that if an environment has a a lot of deterministic players of bounded complexity who win against it, then there exists a simple deterministic player that wins.
\begin{cor}[Many Deterministic Winners]
\label{cor:mdp}
If $2^m$ deterministic players $\r$ of Kolmogorov complexity $\K(\r)\leq n$ win against an environment $\q$, then there exists a deterministic player $\p$ with $\K(\p)\lel n-m+\I(\q;\mathcal{H})$ that wins against $\q$.
\end{cor}
\begin{cor}[Many Probabilistic Winners]
\label{cor:mdp2}
    If $2^m$ probabilistic players $\r$ of Kolmogorov complexity $\K(\r)\leq n$ win against an environment $\q$ with probability $p$, then there exists a deterministic player $\p$ with $\K(\p)\lel n - m -\log p +\I(\q;\mathcal{H})$ that wins against $\q$.
\end{cor}
The following theorem provides bounds using information with the environment. It is a followup to Theorem 2 in \cite{EpsteinBe2011} which is a rearranging of Theorem 2 in \cite{VereshchaginVi04}. 
\begin{thr}[Environment Information]
\label{thr:iq}
    For deterministic environment $\q$, there exists a deterministic agent $\p$ with $\K(\p)\lel \min_{\r\in\mathcal{A}}\I(\r;\q)-\log\mathrm{Win}(\r,\q)+\I(\q;\mathcal{H})$ that wins against $\q$.
\end{thr}
\begin{proof}
For $n\in\N$, let $\mathcal{A}(\q,n)=\{\r:\r\in\mathcal{A},\mathrm{Win}(\r,\q)>2^{-n}\}$. Let $s=\floor{-\log\sum_{\r\in\mathcal{A}(\q,n)}\m(\r)}$. 
So we have that 
\begin{align}
\label{eq:iq1}
\xi(\q)&\lea s+n.
\end{align}
Since $\mathcal{A}(\q,n)$ is lower computable given $s$, $n$, and $\q$, we have that $m(\r)=[\r\in\mathcal{A}(\q,n)]2^s\m(\r)$ is a lower computable semimeasure. So if
$\r\in\mathcal{A}(\q,n)$, then 
\begin{align}
\nonumber
\m(\r|\q,n,s)&\gem m(\r) \eqm 2^s\m(\r) \\
\nonumber
s&\lea \K(\r) -\K(\r|\q,n,s) \\
\nonumber
s&\lel \K(\r) -\K(\r|\q)+\K(n) \\
\label{eq:iq2}
s&\lel \min_{\r\in\mathcal{A}(\q,n)}\K(\r) -\K(\r|\q)+\K(n).
\end{align}
Combining Equations \ref{eq:iq1} and \ref{eq:iq2} and Theorem \ref{thr:nb1}, there is a deterministic agent $\p$ that wins against $\q$, where for all $n$,
\begin{align*}
\K(\p)&\lel \min_{\r\in\mathcal{A}(\q,n)}\I(\r;\q) +n +\I(\q;\mathcal{H})\\
&\lel \min_{\r\in\mathcal{A}}\I(\r;\q) -\log\mathrm{Win}(\r,\q) +\I(\q;\mathcal{H}).
\end{align*}
\end{proof}
\begin{thr}
\label{thr:nb2}
Let $\q$ be a computable probabilistic environment. Let $2^{-s}<\xi(\q)$, $s\in\N$. There is a deterministic agent of complexity $\lel 2s+\I(\q;\mathcal{H})$ that wins with probability $>2^{-s-1}$.
\end{thr}
\begin{proof}
    The proof analogously to that of Theorem \ref{thr:nb1}.
\end{proof}
\begin{cor}
\label{cor:nb3}
Let $\q$ be a computable probabilistic environment. Let $2^{-s}<\xi(\q)$, $s\in\N$. Let $\epsilon\in(0,1)$ be computable. There is a deterministic agent of complexity $\lel 2s+\I(\q;\mathcal{H})+O_\epsilon(1)$ that wins with probability $>\epsilon2^{-s}$.
\end{cor}

How does one interpret these new bounds? Let $\mathcal{B}$ be the set of all total computable probabilistic agents and $\q$ be a deterministic environment. Using some reasoning, one can prove the following inequality
$$
\min_{\p\in\mathcal{B}}\K(\p)-\log\mathrm{Win}(\p,\q)\lel \xi(\q)+\I(\q;\mathcal{H}).
$$
Thus only exotic environments will have an average win rate that is better than the weighted max win rate. However, the bounds proved in this section are more amenable to manipulation, as shown with Corollaries \ref{cor:mdp} and \ref{cor:mdp2} and Theorem \ref{thr:iq}. This also turned out to be the case with G\'{a}cs entropy and Vit\'{a}nyi entropy, where their difference is only over exotic quantum states but G\'{a}cs entropy has many applications in algorithmic physics.

\section{Partial Derandomization}
With game derandomization, a probabilistic player that wins with probability $p$ can be transformed into a deterministic one, with a complexity term of $-\log p$ in the inequality. In this section we show a general tradeoff of bits versus winning probabilities. A probabilistic agent can be transformed into a more successful probabilistic agent, at the cost of its complexity. The theorem of this section puts bounds on this tradeoff. The proof requires using Corollary 7 from \cite{Epstein24}:

\begin{thr}[Monotone EL Theorem for Sets]
\label{thr:monelmod}
For prefix free set $G$, $m=\ceil{-\log \M(G)}$, $n<m$, $n\in \N$, there exists a set $F$ of size $2^n$ where there exists $y\in F$ and $x\in G$ such that $y\sqsupseteq x$ and
$$\K(F) \lel m-n+ \I(G;\mathcal{H})+\K(n).$$   
\end{thr}
\begin{thr}[Partial Derandomization]
\label{thr:partialderandom}
Let $\q$ be a deterministic environment and $\p$ be a probabilistic player that wins with probability $p$. For $s=\ceil{-\log p}+1$, $r<s$, $r\in \N$, there is a probabilistic agent of complexity $\lel s-r+\I(\langle \p,s,\q\rangle;\mathcal{H})+\K(r)$ that wins with probability $>2^{-r}$.
\end{thr}
\begin{proof}
Let $s = \ceil{-\log p}+1$. Using the same reasoning as the proof of Theorem \ref{thr:newproof}, there exists a clopen set $G\subset\IS$ where each element encodes a winning interaction with the environment $\q$, $\p$ can be interpreted as a measure over the Cantor space,  $\p(G)>2^{-s}$, and $\K(G|\p,s,\q)=O(1)$. Using Theorem \ref{thr:monelmod}, one gets a prefix free set $F$, such that there exists $y\in F$ and $x\in G$ such that $y\sqsupseteq x$, $m=\ceil{-\log \M(F)}$, where 
\begin{align}
\nonumber
\K(F) &\lel m - r +\I(F;\mathcal{H}) + \K(r)\\
\nonumber
\K(F) &\lel \K(\p) + s - r +\I(F;\mathcal{H}) + \K(r)\\
\label{eq:partderand1}
\K(F) &\lel \K(\p) + s - r +\I(\langle \p,s,\q\rangle;\mathcal{H}) + \K(r).
\end{align}
Equation \ref{eq:partderand1} is due to Lemma \ref{lmm:consh}. Each such string in $F$ represents a deterministic player. At least one such deterministic player has a winning interaction with $\q$. It is not hard to construct a probabilistic player $\p'$ from $F$ that wins with probability $|F|^{-1}$. Indeed, at every given history $h\in (\N\times \N)^*$ the player $\p'$ chooses action $a\in\N$ with probability $n_a/m_a$, where $m_a$ is the number of deterministic players that are consistent with $h$, and $n_a$ are the number of such agents that choose $n_a$ as their action. If $m_a$ is 0, then $\p'$ can choose any action.
\end{proof}
\begin{cor}[Improved Bounds]
Let $\q$ be a deterministic environment. For $r<\xi(\q)$, $r\in \N$, there is a probabilistic agent $\p$, with $\K(\p)\lel \xi(\q)-r+\I(\q;\mathcal{H}|r)+\K(r)$ that wins with probability $>2^{-r}$.
\end{cor}
\begin{proof}
    From Theorem \ref{thr:partialderandom} and the reasoning of proof of Theorem \ref{thr:nb1}, we have
    \begin{align}
        \nonumber
        \K(\p) &\lel \xi(\q) -r + \I(\langle \xi(\q),\q\rangle;\mathcal{H})+\K(r)\\
        \nonumber
        \K(\p|r) &\lel \xi(\q) -r + \I(\langle \xi(\q),\q\rangle;\mathcal{H}|r)\\
        \label{eq:partderand2}
        \K(\p) &\lel \xi(\q) -r + \I(\langle (\xi(\q)-r),\q\rangle;\mathcal{H}|r)+\K(r)\\
        \nonumber
         &\lel \xi(\q) -r + \I(\q;\mathcal{H}|r)+\K(r).
    \end{align}
    Equation \ref{eq:partderand2} is due to Lemma \ref{lmm:consh}, where $\K((\xi(\q)|r,\xi(\q)-r)=O(1)$.
\end{proof}
\section{Lose/No-Halt Games}
A Lose/No-Halt game is a series of interactions where the player and environment exchange natural numbers. At any given time, the environment can choose to declare that the player has lost. Otherwise, if the game continues forever, the player has won. However Lose/No-Halt games present several barriers in their characterization which the Win/No-Halt games do not have. One such obstacle to derandomizing Lose/No-Halt games is as follows.
\begin{thr}
    There exists a Lose/No-Halt game consisting of a computable environment $\q$, a computable probabilistic player $\p$ that wins with positive probability, and no deterministic player that can beat $\q$.
\end{thr}
\begin{proof}
    The environment $\q$ gives a single action: 0. In addition, the player loses if she players a number other than 0 or 1. We inductively define a computable environment $\q$. At each of the environment turn at $n$, the interaction consists of a bit string $x$ of length $2n-1$. Let $y$ be the odd bits of $x$, i.e. the player's actions. Let $\K_t(c)=\min\{\|p\|:U(p)=c\textrm{ in }\leq t\textrm{ steps}\}$. If $\|z\|-\K_n(z)>d$ for any $z\sqsubseteq y$, then environment $\q$ outputs a loss for the action and the game ends. Otherwise the game continues. This means $\q$ is computable. Thus a player who randomly chooses a bit will output with non zero probability a winning sequence that has less than $d$ randomness deficiency. In addition there is no deterministic player who can output an infinite sequence with bounded randomness deficiency.
 \end{proof}
 \section{Agent Spaces}
 In this section, we revisit agent spaces, introduced in my first paper on algorithmic information theory \cite{EpsteinBe2011}. It also can be seen as a game theoretic variant of \textit{distortion families}, introduced in \cite{VereshchaginVi04}. 
 
 An agent space is a finite or infinite set of (semi)agents. The reason for using agent spaces is that one might want to restrict the agents under consideration. This is an attempt to bridge the gap between results in algorithmic information theory of unrestricted constructs (such as agents in our case) and very restrictive computable models seen in areas such as the Minimum Description Length Principle \cite{Grunwald07}. Frankly, it remains to be seen if such a middle ground exists. We will consider, finite, enumerable, and uncomputable agent spaces. We will leverage the EL Theorem.
 \begin{thr}[EL Theorem, \cite{Levin16,EpsteinEL23}]
  \label{thr:el}
  For finite $D\subset\FS$,\\ $\min_{x\in D}\K(x)\lel -\log\m(D)+\I(D;\mathcal{H})$.
 \end{thr}
 \begin{thr}
 Let $\q$ be a deterministic environment and let $\mathcal{S}$ be an enumerable agent space consisting solely of deterministic agents. Let $D\subseteq\mathcal{S}$ be the set of deterministic agents in the agent space $\mathcal{S}$ that win against $\q$. Note that $D$ may be infinite.
  $$
  \min_{\p\in D}\K(\p)\lel \min_{\p\in D}\I(\p;\langle\q,\mathcal{S}\rangle)+\I(\langle \q,\mathcal{S}\rangle;\mathcal{H}).
  $$
 \end{thr}
 \begin{proof}
 Let $s= \ceil{-\log \m(D)}+1$. Given $\q$, $s$, and a program that enumerates $\mathcal{S}$ consisting of programs of deterministic players, one can create a set $F\subseteq D$ consisting of programs of deterministic agents that win against $\q$ and $\m(F)>2^{-s}$. So $\K(F|\q,\mathcal{S})$. By the EL Theorem \ref{thr:el}, there exists a deterministic agent $\p\in F$ that wins against $\q$ with
 \begin{align}
 \nonumber
     \K(\p) &\lel -\log\m(F) + \I(F;\mathcal{H})\\
     \label{eq:playerspace1}
      &\lel -\log\m(F) + \I(\langle \q,\mathcal{S}\rangle;\mathcal{H})\\
      \nonumber
      &\lel -\log\m(D) + \I(\langle \q,\mathcal{S}\rangle;\mathcal{H}).
 \end{align}
 Equation \ref{eq:playerspace1} is due to Lemma \ref{lmm:consh}. By the definition of $D$, $m(\r)=[\r\in D]\m(\r)2^{s-2}$ is a lower computable semi measure with $\m(\r|\q,\mathcal{S})\m(\r)\gem m(r) $. So if $\r\in D,$ 
 \begin{align*}
 \m(\r|\q,s,\mathcal{S})&\gem \m(\r)2^s\\
 s & \lel \I(\r;\langle \q,\mathcal{S}\rangle).
 \end{align*}
So
$$
\K(\p)\lel \min_{\r\in D}\I(\r;\q,\mathcal{S}) +\I(\langle \q,\mathcal{S}\rangle;\mathcal{H}).
$$
 \end{proof}
 Using the similar reasoning, one can prove the following corollary.

 \begin{cor}
     Let $\q$ be a deterministic environment and let $\mathcal{S}$ be a deterministic agent space. If there are $2^k$ agents $\r\in\mathcal{S}$, with $\K(\r)<r$ that win against $\q$, then there is a deterministic agent $\p\in\mathcal{S}$ that wins against $\q$ and 
     $$
\K(\p)\lel r - k + \I(\langle \q,\mathcal{S}\rangle;\mathcal{H}).
     $$
 \end{cor}

 The results of this section also apply to uncomputable agent spaces consisting of lower computable semi-agents. We won't repeat the theorems, but will define the information of an uncomputable agent space with the halting sequence. For an semi-agent space $\mathcal{S}$, let $\ell[\mathcal[\mathcal{S}]$ be the set of all infinite sequences $\alpha\in\IS$ consisting of encodings $\langle n\rangle$ of numbers $n\in\N$ representing programs to lower computable a semi agent. Furthermore, if there is a semi-agent $\p$ in $\mathcal{S}$ there exists a single encoded number in all such $\alpha$ that lower computes $\p$. The following definition uses Definition \ref{dff:infinf}.
 \begin{dff}[Information, Uncomputable Agent Spaces]
    For agent space $\mathcal{S}$,\\ 
   $\I(\mathcal{S}:\mathcal{H})=\inf_{\alpha\in\ell[\mathcal{S}]}\I(\alpha:\mathcal{H}).$
 \end{dff}

\section{Conclusion}
With this paper, characterizations of the Kolmogorov complexity of deterministic players against deterministic or probabilistic environments that are computable, lower computable, and uncomputable are given. As shown in Section \ref{sec:ue}, the computability of the environment does not play much of a factor in bounds proved. Instead, the key factor is the amount of information the halting sequence has about the environment. There is a general obstacle in derandomizing Lose/No-Halt games, but there might be further assumptions about such games that could be made to circumvent this. This paper proved, for the first time to the author's knowledge, upper bounds on the resource bounded Kolmogorov complexity of players to games. In particular, zero-sum repeated games were studied. One open problem is whether there exists a more straightforward proof that does not use stochasticity, $\Ks$. Such an advancement could be applied to the proofs of other theorems which have $\I(\cdot;\mathcal{H})$ in their inequalities.

\end{document}